\documentclass[11pt, twoside, letterpaper]{amsart}

\usepackage{hyperref}  
\usepackage{latexsym} 
\usepackage{enumerate}
\usepackage{color, xcolor}
\usepackage{amsmath, amsthm, amsfonts, amssymb}
\usepackage{srcltx}
\usepackage[normalem]{ulem} 
 \usepackage{cancel} 
 \usepackage{geometry}
 \usepackage{subcaption}
\usepackage{tikz}
\usepackage{latexsym, amssymb, bbm}
\usepackage{amsmath}
\usepackage{mathrsfs}
\usepackage{amsthm}
\usepackage{verbatim}
\usepackage{fancyhdr}
\usepackage{tikz}
\usetikzlibrary{arrows}
\usepackage{caption}
\usepackage{forest}
\usepackage{graphicx}
\usepackage{stmaryrd}
\usepackage{wasysym}
 
 \makeatletter
\@namedef{subjclassname@2020}{%
  \textup{2020} Mathematics Subject Classification}
\makeatother

\theoremstyle{plain}%
\newtheorem{Theorem}{Theorem}[section] %
\newtheorem{Lemma}[Theorem]{Lemma}

\theoremstyle{definition}%
\newtheorem{Assumption}[Theorem]{Assumption}%
\newtheorem{Remark}[Theorem]{Remark} 
\newtheorem{Example}[Theorem]{Example} 
\newtheorem{Definition}[Theorem]{Definition}

\DeclareMathOperator*{\argmax}{arg\:max}

\def\E{\mathbb{E}}

\def\N{\mathbb{N}}

\def\P{\mathbb{P}}

\def\R{\mathbb{R}}

\newcommand{\cA}{\mathcal{A}}

\newcommand{\cE}{\mathcal{E}}
\newcommand{\cF}{\mathcal{F}}

\newcommand{\cK}{\mathcal{K}}

\newcommand{\cX}{\mathcal{X}}
\newcommand{\cY}{\mathcal{Y}}
\newcommand{\cZ}{\mathcal{Z}}

\def\e{\varepsilon}

\newcommand{\PP}{\mathbb{P}}

\newcommand{\be}{\begin{equation}}
\newcommand{\ee}{\end{equation}}
\newcommand{\bs}{\begin{split}}
\newcommand{\es}{\end{split}}
\newcommand{\ba}{\begin{aligned}}
\newcommand{\ea}{\end{aligned}}
\renewcommand{\[}{\left[}
\renewcommand{\]}{\right]}
\renewcommand{\(}{\left(}
\renewcommand{\)}{\right)}

\renewcommand{\P}{\PP}
\newcommand{\Pas}{\text{$\P$--a.s.}}

\newcommand{\tr}[1]{\textcolor{black}{#1}}
\newcommand{\trr}[1]{\textcolor{black}{#1}}

\def\tcb{\textcolor{black}}

\renewcommand{\epsilon}{\varepsilon}
\newcommand{\nn}{\nonumber}

\renewcommand{\bar}[1]{\overline{#1}}
\begin{document}

\title[\tcb{The value of partial information}]{\tcb{The value of partial information}}

\author{Philip A. Ernst}

\address{
Philip A. Ernst, 
Department of Mathematics, Imperial College London,  London, SW7 2AZ
United Kingdom}
\email{p.ernst@imperial.ac.uk
}\thanks{\tcb{The first-named author acknowledges support from The Royal Society Wolfson Fellowship.}}

\author{Oleksii Mostovyi}\thanks{The \tcb{second-named} author has been supported by the National Science Foundation under grant  No.   DMS-1848339. 
Any opinions, findings, and conclusions or recommendations expressed in this material are those of the author and do not necessarily reflect the views of the National Science Foundation.
}
\address{Oleksii Mostovyi, Department of Mathematics, University of Connecticut,  Storrs, CT 06269, United States}%
\email{oleksii.mostovyi@uconn.edu}%


\date{\today}%


\maketitle

\begin{abstract} 
We investigate a pricing rule that is applicable for streams of income or contingent claim liabilities and study how this rule changes under additional insider-type information that an investor might obtain. Considering a model where the risky asset might have jumps, we obtain an explicit form of the associated state price density \tcb{for the three different types of agents considered in \cite{Philip1}: one who has no information
about the jumps, one who knows in advance exactly when the each jump will occur, and one who
has no information about the size of the jumps but has partial information about the size of each
jump}. \tcb{For each of these agents}, we provide characterizations of the pricing rule and establish a representation formula, \trr{allowing us to quantify the value of partial information for streams of labor income \tcb{or} contingent claim liabilities}. Our work is motivated by finding and characterizing a pricing rule that, \tcb{both with or without} partial information about jumps, \trr{assigns different values of information for} different income streams or contingent claim liabilities.
\end{abstract}  
\section{Introduction}\label{Intro}
Insider trading is a very active area of financial mathematics and related fields. 
This topic has \tcb{influenced virtually every subfield} of mathematical finance, including arbitrage theory, pricing and hedging, \tcb{characterizations of equilibria}, optimal investment, etc. For references on these topics, we refer the reader to \cite{KarPik}, \cite{ElliotJeanblanc}, \cite{MartinII2}, \cite{Fabrice03}, \cite{higa}, \cite{imkellerAsym},  \cite{CampiII}, 
  \cite{ImkellerAI}, \cite{umut1}, \cite{KostasInformation}, \cite{Michael1}, \cite{HaoIT}, \cite{aksamit1},  \cite{umut2}, \cite{kosBeaFont16},  
 \cite{moniqueAksamit},  \cite{aksamit4}, \cite{Philip1}, \cite{aksamit3}, \cite{aksamit2},  \cite{PaoloLP}, \cite{Philip2}, \cite{claudioIA}, \cite{ScottAI1}, \cite{umut3}, \cite{ScottAI2}, \cite{Choi2023}, \cite{Shi2024}, and \cite{Scott2024}.

\tcb{The present paper is in part motivated by the recent work of} \cite{Philip1}. \tcb{There, an approach for modeling insider information is proposed which allows for the incorporation of jumps in the underlying stock price process,} \tcb{without sacrificing analytic tractability}. The value of the signal in \cite{Philip1} is represented via a change in the value function of the rational investor with and without insider information. \tcb{In particular, \textit{the value of the signal in \cite{Philip1} does not depend}\textit{on a stream of labor income or contingent claim liabilities.}}

\indent \tcb{The present paper fills this aforementioned void by considering the problem of pricing streams of liabilities or labor income for two types of partially informed insider traders, which is then compared with a pricing rule associated with an uninformed trader.} \tcb{As in \cite{Philip1}}, the first type of insider trader, \tcb{who we call} the insider of the ``first kind,'' knows the times of the next jump of the risky asset but does not know the size of the next jump. The second type of the inside trader, who we call an insider of the ``second kind,'' does not know the size of the next jump but has partial information about the size of the next jump; that is, he observes the signal $\eta = \xi + \varepsilon$, where $\xi$ is the jump of the cumulative return of the risky asset, and $\e$ is a random variable representing noise, \tcb{ assumed} Gaussian and independent of everything else. 
 
It is well-known that such unbounded \tcb{ (from above and below)} streams of labor income/contingent claim liabilities are \tcb{more challenging} to work with, \tcb{in particular because} 
  the notion of admissible wealth processes has to be adjusted accordingly. Following \cite{DS1997}, we incorporate the notion of acceptability, and we show that under no-arbitrage-type and a version of the super-replicability conditions, the pricing problems \tcb{(both with and without additional information about the signal) admit a unique solution, for which we identify a representation formula.}\\
\indent \tcb{The present paper's key contributions are threefold. Firstly, we identify and characterize a pricing rule which
depends on a particular stream of labor income or contingent claim liabilities and which assigns different values of partial information to different streams. Secondly, we develop a framework which allows for an infinite time horizon such that no arbitrage or other technical issues arise in the presence of extra information, which is indeed a very delicate topic in general (see \cite{kosBeaFont16} and \cite{moniqueAksamit}). Our framework allows for streams of liabilities or labor income to be unbounded from either above or below, enabling greater generality. Thirdly, we present examples of special interest \trr{that contain explicit computations for intricate streams and show the existence} of information-invariant \trr{ones}.}

The remainder of this paper is organized as follows. In Section \ref{secModel}, we describe the model and provide existence and uniqueness results for the pricing rule \tcb{developed}. In Section \ref{secKnowsNothing}, we provide an explicit formula for the state price density associated with the pricing rule in the case \tcb{where} the investor has no information about the jumps of the risky asset. Section \ref{secKnowsTimes} contains the derivation and representation of the pricing rule in the case \tcb{where} the investor knows the time of the next jump of the risky asset. In Section \ref{secKnowsSize}, we derive the corresponding state price density for the investor who does not know the times of the jumps, but has some partial information about the size of the next jump. We conclude with Section \ref{secDiscussion}, which discusses explicit forms of the pricing rules. 
 \section{Problem formulation and characterizations of the pricing rule}\label{secModel}

\tcb{We} consider an infinite horizon continuous-time problem of investment and consumption. For this, \tcb{we assume} a complete stochastic basis $\(\Omega, \cF, (\cF_t)_{t\geq 0}, \P\)$, where the filtration satisfies the usual conditions and where $\cF_0$ is trivial. We suppose that there are two traded assets:
  a risk-free bank account with a constant interest rate $r>0$ 
   and   a risky asset whose returns are given by 
 \be\label{defR}
 d{\sf R}_t = \mu d t + \sigma d W_t + \int(e^x - 1)n(dt, dx), \quad t> 0,\quad {\sf R}_0 = 0,
 \ee
 where $W$ is a standard one-dimensional Brownian motion, $\sigma>0$ is the constant-valued volatility, $\mu$ is the constant-valued drift, and $n$ is a Poisson random measure independent of $W$, such that 
 $$\E[n(dt, dx)] = \lambda dt p(x) dx,$$
 where $\lambda$ is the rate at which jumps occur and $p$ is the density of the jump of the cumulative return process. We suppose that, for some constants $m\in\R$ and $v>0$, $p(x) = \frac 1{\sqrt{2\pi v}}e^{-\frac{(x - m)^2}{2 v}}$, $x\in\R$, which is the normal density with parameters $(m, v)$. 
 We remark that the \tcb{above model dates} back to \cite{MertonJumps}. 
 The evolution of the wealth of the agent is given by 
\be\label{w0}
dw_t =(rw_t - c_t)dt + \theta_t\( (\mu - r) dt + \sigma dW_t + \int(e^x - 1)n(dt, dx)\), \quad t> 0, \quad w_0 = x,  
\ee
 where $x$ is an initial wealth,   $\theta$ is a predictable and ${\sf R}$-integrable process specifying the amount of wealth invested in the risky asset, and $c$ is a consumption rate, which we assume to be optional. An optional consumption rate $c\geq 0$ is admissible from the initial wealth $x$ if there exists a predictable and ${\sf R}$-integrable process $\theta$ such that 
the associated wealth process given \tcb{in} \eqref{w0} is nonnegative. We denote the set of admissible consumption streams from the initial wealth $x$ by $\cA(x)$, $x\geq 0$, and the set of nonnegative wealth processes $w$ given by \eqref{w0} starting from the initial wealth $x\geq 0$ and associated with $c\equiv 0$ by $\cX(x)$, $x\geq 0$.
We also fix a power utility of the form
\be\label{defU}
U(x) := \frac{x^{1 - R}}{1 - R},\quad x>0,
\ee
\tcb{with} relative risk aversion $R>0$ and $R\neq 1$.  
With these preliminaries \tcb{in hand, we can now} specify the value function
\be\label{primalProblem}
u(x) = \sup\limits_{c\in\cA(x)}\E\[\int_0^\infty e^{-\rho t}U(c_t) dt|w_0 = x\],\quad x>0,
\ee
where $\rho$ is the time discounting factor. 
 
\tcb{As noted in the introductory section, the purpose of the present paper} is to price a stream of labor income/random endowment \tcb{for the three different types of agents considered by \cite{Philip1}}:
\begin{enumerate}
\item The agent has no prior knowledge about when the jumps occur nor of their magnitudes;
\item The agent knows precisely the time of the next jump, but not the magnitude;
\item The agent knows nothing about the time of the next jump but sees the signal
$$\eta = \xi + \e,$$
where $\xi$ is the jump in the cumulative return at the next jump, and $\e$ is an independent $N(0, v_\e)$ random variable. 
\end{enumerate}
\indent \indent \tcb{We shall assume that the labor income} is an optional process $e$. If the agent has $q\in\R$ units of $e$, the evolution of his wealth is given by 
\be\label{w}
dw_t =(rw_t - c_t +  {q e_t})dt + \theta_t\( (\mu - r)dt + \sigma  dW_t + \int(e^x - 1)n(dt, dx)\), \quad t> 0, \quad w_0 = x, \quad q\in\R.
\ee
 Similarly, for a pair $(x,q)\in (0,\infty)\times \R$, we say that an optional consumption stream $c\geq 0$ is admissible from the initial wealth $x$ and the number of units of income process $q$ if there exists a predictable and ${\sf R}$-integrable process $\theta$, such that
\tcb{(i) the associated wealth process given by \eqref{w} is nonnegative at all times $\Pas$ and (ii) the self-financing wealth process} $X= w+ \int_0^\cdot (c_s -qe_s)ds$ is \tcb{``acceptable''} in the sense of \cite{DS1997}. \tcb{This means that it} can be written as a difference of two nonnegative self-financing wealth processes $X' - X''$, where $X''$ is maximal in $\cX(X''_0)$. \tcb{We remark that the notion of acceptability introduced in \cite{DS1997} has played} an important role in optimal investment problems with labor income or random endowment, see \cite{HK}, \cite{Pietro2}, \cite{MostovyiEnd}, and \cite{MostovyiSirbuUnified}.\\
\indent \tcb{Some further notation is now in order.} We denote the set of admissible consumption streams from the initial wealth $x$ and with the number of units of income process $q$ by $\cA(x,q)$, $(x,q)\in (0,\infty)\times \R$. \tcb{Note that} $\cA(x,q)$ \tcb{may} be empty for some $(x,q)\in\R^2$. Following \cite[Definition on page 157]{KarKar21}, \tcb{we recall the definition of an indifference value.}
 \begin{Definition}\label{defInd}
  Let $x>0$ be fixed. A number $p$ is an indifference value for $e$ at $x$, if 
 $$\E\[\int_0^\infty e^{-\rho t}U(c_t  ) dt|w_0 = x\] \leq u(x),\quad for~every\quad q\in\R\quad and\quad c\in\cA(x - qp, q).$$
 \end{Definition}
 \begin{Remark}\label{remPindepx}
 In view of $U$ defined in \eqref{defU} having a power form, one can see that the indifference price does not depend on $x>0$. 
 \end{Remark}
 Let $\cX(x)$ be the set of nonnegative wealth processes \tcb{of} the form \eqref{w0} associated with $c \equiv  0$.
 With $B_t = e^{rt}$, $t\geq 0$, we define the set of local martingale deflators for the discounted risky asset $\frac{\cE(R)}{B}$ by
\be\nn\ba
\cZ(y) =& \left\{ Z>0: ~ Z\frac{X}{B}~is~a~\P\text-local~martingale~ for~every~X\in\cX(1)\right\},\quad y\geq 0.
\ea\ee
\tcb{In view} of \eqref{defR}, particularly \tcb{when} $\sigma>0$, we observe that 
\be\label{NUPBR}
\cZ(1)\neq \emptyset.
\ee
\tcb{The condition in} \eqref{NUPBR} is equivalent to the no unbounded profit with bounded risk (NUPBR) condition introduced in \cite{KarKar07} (via the results in \cite{TakaokaSchweizer} and \cite{KabKarSong16}).  
\begin{Remark}
Without the condition in \eqref{NUPBR}, the \tcb{optimization problem in} \eqref{primalProblem} is not well-posed. In the context of optimal investment from terminal wealth, illuminating examples are given in 
 \cite[Proposition 4.19]{KarKar07}. Furthermore, we need to suppose that \eqref{NUPBR} holds in the \tcb{present paper}. An example of the model where \eqref{NUPBR} fails is given by 
 the riskless asset in \eqref{defR} as above and the risky asset having the return
 $$rt + N_t, \quad t\geq 0,$$
 where $N$ is a Poisson process with intensity $1$. 
\end{Remark}

\tcb{We proceed to denote} the set of  state price density processes 
 $$\cY(y) = \left\{ Y = \frac{Z}{B}:~Z\in\cZ(y)\right\},\quad y\geq 0,$$
and consider an income stream $e = (e_t)_{t\geq 0}$ such that 
\be\label{superRep}
\sup\limits_{Y\in \cY(1)} \E\[ \int_0^\infty Y_t |e_t| dt\]<\infty.
\ee
\begin{Remark}
One can see that \eqref{superRep} holds if there exist constants $C>0$ and $r'\in [0, r)$, such that
$$|e_t(\omega)|\leq Ce^{r't} , \quad (dt\times\P)\text-a.e. $$
\end{Remark}
For $q := \frac{1 - R}{R}$, we introduce $V(y) := \sup\limits_{x>0}\(U(x) - xy\) = \frac{y^{q}}{-q}$, $y>0$, and set the dual value function as
\be\label{dualProblem}
v(y) := \inf\limits_{Y\in \cY(y)}\E\[\int_0^\infty e^{-\frac \rho Rt}V(Y_t)dt\],\quad y>0.
\ee
Now, suppose that 
\be\label{finValue}
u(z)>-\infty \quad and \quad v(z)<\infty,\quad z>0.
\ee
\begin{Remark}
If $R\in(0,1)$, $U$ is positive-valued and therefore $u(x)>-\infty$ \tcb{for all} $x>0$.
If $R>1$, a sufficient condition for $u(x)>-\infty$, $x>0$, is
\be\label{7291}
\rho \geq  (1 -R)r.
\ee
\tcb{This result follows from} \cite[Lemma 4.2]{MostovyiNec}. \tcb{Namely,} for every constant $\delta\in(0, r)$, $\bar c_t := \delta e^{(r-\delta) t}$, $t\geq 0$, is an element of $\cA(1)$, and  
 $$\E\[\int_0^\infty e^{-\rho t} U(\bar c_t) dt\] = \frac{\delta^{1-R}}{1-R}\int_0^\infty e^{-\rho t} e^{(1-R)(r - \delta)t} dt>-\infty,$$ 
 where the inequality follows from \eqref{7291}. Such a $\bar c$ \tcb{will give} a finite lower bound for $u$. 
 
\tcb{The condition $v(z)<\infty,$ $ z>0$, when $R>1$ holds under \eqref{NUPBR}} \tcb{since the objective in \eqref{dualProblem} is negative-valued}. If $R\in(0,1)$, a necessary and sufficient condition for $v(z)<\infty,$ $ z>0$, is the existence of one-state price density $Y\in\cY(1)$ such that 
 $$\E\[\int_0^\infty e^{-\frac \rho Rt}V(Y_t)dt\]<\infty.$$
\end{Remark}
\tcb{To obtain uniqueness} of the indifference price, we shall need \tcb{Assumption \ref{asUnique} below.}
\begin{Assumption}\label{asUnique}
Suppose there exists an optional process $\tilde c$, such that 
$$\tilde c\geq |e|,\quad (dt\times \P)\text-a.e.,$$
and, for the the minimizer to \eqref{dualProblem} at $y=1$,  {$\hat Y\in\cY(1)$}. \tcb{Further,}
\be\label{indUnique}
\E\[ \int_0^\infty \hat Y_t \tilde c_t dt\] = \sup\limits_{Y\in \cY(1)} \E\[ \int_0^\infty Y_t\tilde c_t  dt\]<\infty.
\ee
\end{Assumption}
\begin{Remark}
 Condition  $\hat Y\in\cY(1)$ is included in Assumption \ref{asUnique}. It will be used in \eqref{7293} below.  
\end{Remark}
\begin{Remark}\label{remYmart}[On the sufficient conditions for Assumption \ref{asUnique}]
For the uniqueness of the indifference price represented by \eqref{iRep}, Assumption \ref{asUnique} holds if 
\begin{enumerate}[(a)]\item 
there exists a {\it deterministic} consumption stream $\tilde c_t$, $t\geq 0$, such that 
$$\tilde c_t \geq |e_t|,\quad (dt\times \P)\text-a.e.,$$
\item 
which satisfies $$\int_0^\infty \frac{\tilde c_t }{B_t }dt <\infty,$$
\item 
and the minimizer to  \eqref{dualProblem} at $y=1$, $\hat Y$, is  such that $\hat YB = (\hat Y_t B_t)_{t\geq 0}$ is a $\P$-martingale.
\end{enumerate}
Then, under the remaining assumptions of Theorem \ref{thmIRep}, the proof goes through, where the key step \tcb{follows} via Tonelli's theorem, 
\be\nn\ba
\E\[\int_0^\infty\tilde c_t\hat  Y_tdt\] = \int_0^\infty \tilde c_t \E [\hat Y_t ] dt  = \int_0^\infty \frac{\tilde  c_t}{B_t}  dt <\infty.
\ea\ee
The inequality  holds by condition ${\rm (b)}$, 
and for every $Y\in\cY(1)$, we have
\be\nn\ba
\E\[\int_0^\infty\tilde c_t   Y_tdt\] = \int_0^\infty \tilde c_t \E [  Y_t ] dt  \leq \int_0^\infty \frac{\tilde  c_t}{B_t}  dt = \E\[\int_0^\infty\tilde c_t\hat  Y_tdt\].
\ea\ee
\end{Remark}

\tcb{We now proceed to state Theorem \ref{thmIRep}. Its proof shall follow the proof of Lemma \ref{lem961}.}
\begin{Theorem}\label{thmIRep}
Let us consider the market given as above, where there is a riskless security with a constant interest rate and a risky \tcb{asset}, whose return $R$ is given by \eqref{defR} with  $\sigma>0$\footnote{This ensures that there is no arbitrage.}. Suppose that \eqref{superRep}  and \eqref{finValue} hold. 
 Then, for every $x>0$, the indifference value $p(x)$ of the income stream $e$ exists and \tcb{does not depend on $x>0$. Further,}
\be\label{iRep}
p(x) = \E\[ \int_0^\infty {\hat Y_t} e_t dt\],\quad x>0,
\ee
is an indifference price of $e$, where $\hat Y$ is the minimizer to \eqref{dualProblem} at $y=1$. If Assumption \ref{asUnique} \tcb{also} holds, then the indifference price of $e$ is unique and is given by \eqref{iRep}. 
\end{Theorem}
\begin{Remark} Theorem \ref{thmIRep} can be extended to the case when $\cF_0$ is not trivial but is a sigma-algebra generated by a random variable such that the conditional independence of increments denoted in \cite[Section II.6]{JS} holds. 
Then, the indifference value is an $\cF_0$-measurable random variable that is given via a conditional expectation with respect to $\cF_0$, that is
\be\label{iRepCond}
p= \E\[ \int_0^\infty {\hat Y_t} e_t dt|\cF_0\]. 
\ee
As in the unconditional case, $p$ does not depend on the initial wealth, which, in this case, is a nonnegative $\cF_0$-measurable random variable. $\hat Y$ is the dual optimizer with the initial value $1$, that is, starting from the dual problem, where \tcb{the initial values can be the nonnegative $\cF_0$-measurable random variable $Y_0$}. \tcb{Adopting} the standard notation $\frac 00:=0$, and denoting the optimizer by $Y$, we can represent $p$ in \eqref{iRepCond} via $\hat Y = \frac{Y}{Y_0}$.
\end{Remark}
\tcb{The proof of Theorem \ref{thmIRep} requires Lemma \ref{lemK} and Lemma \ref{lem961} below. We proceed with Lemma \ref{lemK}.}
\begin{Lemma}\label{lemK}
Under the conditions of Theorem \ref{thmIRep}, the set 
\be\label{defK}
\cK:=\{(x,q)\in\R^2:~\cA(x,q)\neq \emptyset\}
\ee
 is a nonempty convex cone in $\R^2$ such that 
 \be\label{Kint}
 \{(x,0):~x> 0\}\in int\cK.
 \ee
\end{Lemma}
\begin{proof}
The fact that $\cK$ is a convex cone follows from its definition and the definition of the sets $\cA(x,q)$. In view of \eqref{superRep}, using localization and \tcb{employing the construction of} \cite[proof of Lemma 4.2]{MostovyiNec} with $\tilde x: = \sup\limits_{Y\in\cY(1)}\E\[\int_0^\infty Y_t |e_t|dt \]$, one can show that there exists 
$X\in\cX(\tilde x)$ such that 
\eqref{w} holds for $c\equiv 0$, every $q\in[-1,1]$, and $$w = X - q\int_0^\cdot e_tdt.$$ 
Next, we observe that by \cite[Remark 3.4]{HKS05} \tcb{that} the assertions of \cite{DS1997} (\tcb{in particular}, \cite[Corollary 2.6]{DS1997}) hold without the \tcb{assumption of local boundedness of the risky asset. Instead}, it is enough to suppose that \tcb{there exists a ``separating measure'' (\cite{Delbaen-Schachermayer1998})} for the risky asset. This assumption holds in the present \tcb{setting}.
Thus, from \cite[Corollary 2.6]{DS1997}, we deduce that $X$ constructed above can be chosen to be maximal. Consequently, $\cK\neq \emptyset$ and $(\tilde x, q)\in\cK$, $q\in[-1,1]$. Moreover, \eqref{w} implies that, for every $\alpha\geq 0$, we have
$$(\alpha\tilde x, \alpha q)\in\cK, \quad q\in[-1,1],$$
which \tcb{in turn} implies \eqref{Kint}.  This completes the proof of the lemma.

\end{proof}
\begin{Lemma}\label{lem961}
Under the conditions of Theorem \ref{thmIRep}, let $(x,q)\in\R^2$ be such that $\cA(x,q)\neq \emptyset$. Then, for the income stream $e$ and \tcb{for} every consumption process $c\in\cA(x, q)$, 
\be\label{961}
 \E\[ \int_0^\infty  \hat Y_t  c_t  dt\]\leq x +  q\E\[ \int_0^\infty  \hat Y_t  e_t  dt\],
\ee
  where $\hat Y$ is the minimizer to \eqref{dualProblem} at $y=1$.
\end{Lemma}
\begin{Remark}
Under the continuity of the stock price process assumption, a closely related result is contained in \cite[Exercise 3.55\tr{, page 156}]{KarKar21}. 
\end{Remark}
\begin{proof}[Proof of Lemma \ref{lem961}]
For a given $(x,q)\in\cK$, where $\cK$ is defined in \eqref{defK}, \tcb{and for} $\cA(x,q)\neq \emptyset$, let us fix an arbitrary $c\in\cA(x,q)$. Let $X$ be a self-financing acceptable wealth process starting from $x$ such that 
\be\label{962}
X_t + q\int_0^t e_s ds \geq \int_0^tc_sds,\quad t\geq 0,\quad \P\text-a.s.
\ee
Let $$\tilde x: = \sup\limits_{Y\in\cY(1)}\E\[\int_0^\infty Y_t |e_t|dt \] =  \sup\limits_{Z\in\cZ(1)}\E\[\int_0^\infty Z_t \frac{|e_t|}{B_t}dt \] .$$
Then, from \eqref{superRep}, using
 \cite[Lemma 1]{MostovyiNUPBR}, 
  we deduce that there exists a self-financing nonnegative wealth process ${\tilde X}\in\cX(\tilde x)$ such that 
\be\label{972}
\frac{{\tilde X}_t}{B_t} \geq \int_0^t \frac{|e_s|}{B_t} ds,\quad t\geq 0, \quad \P\text-a.s.
\ee
\tcb{We proceed by denoting}
$${\cE'}: = \int_0^\cdot \frac{|e_t|}{B_t}dt, 
\quad 
{\bar\cE}: = \int_0^\cdot \frac{e_t}{B_t}dt\quad {\rm and}\quad \bar C: = \int_0^\cdot \frac{c_t}{B_t}dt,
$$
\tcb{Employing the} integration by parts formula in \cite[Proposition I.4.49(a)]{JS} and a change of num\'eraire argument
, we deduce from \eqref{962} that there exists a self-financing wealth process $\bar X$ starting from $x$ such that 
\be\label{971}
 \frac{\bar X_t}{B_t}+ q\bar \cE_t\geq \bar C_t,\quad t\geq 0,\quad \P\text-a.s.
\ee
Therefore, via \eqref{972} and  \eqref{971}, we obtain 
$$0\leq \bar C_t\leq \frac{\bar X_t}{B_t}+ q\bar \cE_t\leq \frac{\bar X_t}{B_t}+ |q| \cE'_t\leq \frac{\bar X_t +|q| {\tilde X}_t}{B_t},\quad t\geq 0, \quad \P\text-a.s.$$
Here, we see that $\bar X +|q|  {\tilde X}$ is a nonnegative self-financing wealth process starting from $x + |q|\tilde x$, that is  
$$\tilde X':=\bar X +|q|  {\tilde X}\in\cX(x + |q|\tilde x).$$
 
It now follows from \cite[Theorem III.29\tr{, page 128}]{Pr} that 
\be\label{963}
{\bar\cE}_{-}\cdot (\hat YB)\quad {\rm and}\quad \bar C_{-}\cdot (\hat YB)\quad{\rm are~\P\text-local~martingales.}
\ee
Furthermore, it follows from the definition of $\cY(1)$ and from Assumption \ref{asUnique} (which, in particular, requires $\hat Y$ to be an element of $\cY(1)$), that ${\tilde X}Y$ and $(X + |q|{\tilde X})\hat Y = \tilde X'\hat Y$ are local martingales as $ \tilde 
X'\in\cX(x + |q|\tilde x)$ and $  {\tilde X} \in\cX(\tilde x)$. Therefore, there exists a localizing sequence of stopping times $\tau_n$, $n\geq 1$, for 
$${\tilde X}\hat Y ,\quad \tilde X'\hat Y,\quad {\bar\cE}_{-}\cdot (\hat YB),\quad {\rm and}\quad {\bar C}_{-}\cdot (\hat YB).$$
As a consequence, we obtain
\be\label{964}
\E\[ {\tilde X}_{\tau_n}  \hat Y_{\tau_n}\] = \tilde x,\quad n\geq 1,
\ee
and
\be\label{965}
\E\[ \tilde X'_{\tau_n} \hat Y_{\tau_n}\] = 
\E\[ (\bar X_{\tau_n} +|q|{\tilde X}_{\tau_n})\hat Y_{\tau_n}\] = x + |q|\tilde x,\quad n\geq 1.
\ee
From \eqref{964} and \eqref{965}, we deduce that
\be\label{966}
\E\[\bar X_{\tau_n} \hat Y_{\tau_n}\] = x,\quad n\geq 1.
\ee
From \eqref{971}, we obtain
\be\label{967}\(\frac {\bar X_{\tau_n}}{B_{\tau_n}} + q\bar \cE_{\tau_n} \)\(\hat Y_{\tau_n} B_{\tau_n}\) \geq \bar C_{\tau_n}\(\hat Y_{\tau_n} B_{\tau_n}\),\quad n\geq 1,\quad \P\text-a.s.
\ee
Using the integration by parts formula in \cite[Proposition I.4.49(a)]{JS}, we restate 
\eqref{967} as
\be\label{968}\bar X_{\tau_n}\hat Y_{\tau_n} + q\bar \cE_{-}\cdot \(\hat YB\)_{\tau_n} +q\int_0^{\tau_n} \hat Y_se_sds \geq \bar C_{-}\cdot \(\hat YB\)_{\tau_n} +\int_0^{\tau_n} \hat Y_sc_sds,\quad n\geq 1,\quad \P\text-a.s.
\ee
Taking \tcb{expectations} and using \eqref{966} and the martingale property of $\cE_{-}\cdot \hat Y_{\cdot \wedge \tau_n}$ and $C_{-}\cdot \hat Y_{\cdot \wedge \tau_n}$, we deduce from \eqref{968}  that 
\be\label{969}
x + q\E\[ \int_0^{\tau_n} \hat Y_se_sds\] \geq \E\[\int_0^{\tau_n} \hat Y_sc_sds\],\quad n\geq 1.
\ee
\tcb{Finally}, taking the limit as $n\to\infty$, and (i) \tcb{invoking dominated convergence} (in view of \eqref{superRep}) on the left-hand side of \eqref{969} and (ii) \tcb{applying monotone convergence} on the right-hand side of \eqref{969}, we conclude that
\be\nn
x + q\E\[ \int_0^{\infty} \hat Y_se_sds\] \geq \E\[\int_0^{\infty} \hat Y_sc_sds\],
\ee
which is \eqref{961}. This completes the proof. 

\end{proof}
\tcb{With the proofs of Lemma \ref{lemK} and Lemma \ref{lem961} in hand, we proceed to prove Theorem \ref{thmIRep}.}
\begin{proof}[Proof of Theorem \ref{thmIRep}]
Let us fix $x=-v'(1)$ and consider an arbitrary $q\in\R$, such that $\cA(x-qp(x), q)\neq \emptyset$, with $p(x)$ given in \eqref{iRep}. Next, let us fix an arbitrary $c\in\cA(x-qp(x),q)$. Using Lemma \ref{lem961}, we deduce that 
\be\label{7293}
 \E\[ \int_0^\infty  \hat Y_t  c_t  dt\]\leq x - qp(x) +  q\E\[ \int_0^\infty  \hat Y_t  e_t  dt\] = x.
\ee
Further, \cite[Theorem 3.2]{MostovyiNec} enables us to conclude that 
\be\label{7292}
u(x) - x  = v(1).
\ee
\tcb{The conjugacy between $U$ and $V$, in conjunction with} \eqref{7293} and \eqref{7292}, allow us to obtain
\be\nn\ba
\E\[ \int_0^\infty e^{-\rho t}U\(c_t \)dt\]&\leq \E\[ \int_0^\infty e^{-\rho t}V (\hat Y_t e^{\rho t}  )dt\] +  \E\[ \int_0^\infty  \hat Y_t  c_t  dt\] \\
&= v(1)+  \E\[ \int_0^\infty  \hat Y_t   c_t  dt\] \\
&= u(x) - x +  \E\[ \int_0^\infty  \hat Y_t   c_t   dt\] \\
&\leq u(x)   -x+x  = u(x).
\ea\ee
Since $q$ and $c$ are arbitrary, we deduce from Definition \ref{defInd} that $p(x)$ as given in \eqref{iRep} is an indifference price for the stream of income $e$. 
The scaling argument (as in Remark \ref{remPindepx}) implies that  \eqref{iRep} holds for every $x>0$.

To prove the uniqueness of the indifference price for $e$ given by \eqref{iRep}, we consider $\pi < p(x)$. \tcb{Suppose that there is a} consumption stream $\tilde c\geq  |e|$, $(dt\times\P)\text-$a.e., satisfying  Assumption \ref{asUnique}. Then, it follows from \eqref{superRep} and \cite[Lemma 4.2]{MostovyiNec} that $\tilde c\in\bigcup\limits_{x\geq 0}\cA(x)$. \tcb{We proceed to denote}
\be\label{defTildex}
\tilde x : = \E\[\int_0^\infty \tilde c_t \hat Y_tdt\],
\ee
 which is finite by Assumption \ref{asUnique}, \tcb{and} which additionally asserts that $\tilde x = \sup\limits_{Y\in\cY(1)}\E\[\int_0^\infty \tilde c_t \hat Y_tdt\]$. Therefore, via \cite[Lemma 4.2]{MostovyiNec} (and localization), we deduce that $\tilde c\in\cA(\tilde x)$. As a consequence, $$\tilde c + e\in \cA(\tilde x, 1).$$ 
 \medskip
\noindent Let $q_n$, $n\in \N$, be a sequence of strictly positive numbers decreasing to $0$ and such that $q_n\leq \frac {x}{\tilde x+ |\pi|}$, $n\in\N$, where $x=-v
'(1)$, which is the same strictly positive constant as in the previous paragraph. With $\hat c(x - q_n(\tilde x + \pi))$ denoting the optimizers to \eqref{primalProblem} at $x - q_n(\tilde x + \pi)>0$, $n\in\N$, we set 
\be\label{defcn}
c^n: = \hat c(x - q_n(\tilde x + \pi)) + q_n (\tilde c + e),\quad n\in\N. 
\ee
In view of Assumption \ref{asUnique}, we have that $c^n\in\cA(x - q_n\pi, q_n)$, $n\in\N.$
\tcb{Employing the concavity of $U$, we obtain}
\be\nn\ba
\E\[\int_0^\infty e^{-\rho t}U(c^n_t )dt \] 
&\geq \E\[\int_0^\infty e^{-\rho t}U(\hat c_t(x - q_n(\tilde x + \pi)))dt \] +q_n\E\[\int_0^\infty  e^{-\rho t} (\tilde c_t + e_t) U'(c^n_t  )dt \]\\
&=  u(x - q_n(\tilde x + \pi))  +q_n\E\[\int_0^\infty  e^{-\rho t} (\tilde c_t + e_t)  U'(c^n_t  )dt \],\quad n\in\N.
\ea\ee
Consequently, for $y = u'(x) = 1$, we get
\be\label{7281}\ba
\liminf\limits_{n\to \infty}\frac{\E\[\int_0^\infty e^{-\rho t}U(c^n_t )dt \] - u(x)}{q_n}&\geq 
-y(\tilde x + \pi) + \liminf\limits_{n\to \infty}\E\[\int_0^\infty  e^{-\rho t} (\tilde c_t + e_t) U'(c^n_t  )dt \]\\
&= 
- (\tilde x + \pi) + \liminf\limits_{n\to \infty}\E\[\int_0^\infty  e^{-\rho t} (\tilde c_t + e_t)  U'(c^n_t  )dt \].
\ea\ee
In view of the power utility structure in the objective of \eqref{primalProblem}, one can see that the optimizers to \eqref{primalProblem} associated with different initial wealths satisfy $\hat c(x) = x\hat c(1)$, $(dt\times \P)\text-a.e.$, $x>0$. Therefore, using the $(dt\times \P)\text-$a.e. non-negativity of $ (\tilde c_t(\omega) + e_t(\omega)) U'(c^n_t(\omega) + q_n e_t(\omega))$ \tcb{for all $n\in\N$, and employing} Fatou's lemma, we can further bound the right-hand side of \eqref{7281} from below by 
\be\label{7294}\ba
&\liminf\limits_{n\to \infty}\E\[\int_0^\infty  e^{-\rho t} (\tilde c_t + e_t)  U'(c^n_t )dt \]\geq 
 \E\[\int_0^\infty  e^{-\rho t} (\tilde c_t + e_t) U'(\hat c_t(x))dt \].
 \ea\ee
 Now, using \cite[Theorem 3.2]{MostovyiNec}, we deduce that for $y = u'(x) = 1$, 
 \be\label{7295}\ba
 &\E\[\int_0^\infty  e^{-\rho t} (\tilde c_t + e_t) U'(\hat c_t(x))dt \]=\E\[\int_0^\infty    (\tilde c_t + e_t) y\hat Y_t  dt \] =\E\[\int_0^\infty    (\tilde c_t + e_t) \hat Y_t  dt \]. 
 \ea
 \ee
 Furthermore, from \eqref{iRep} and \eqref{defTildex}, we obtain 
 \be\label{7296}
 \E\[\int_0^\infty    (\tilde c_t + e_t)  \hat Y_t dt \] =   \tilde x + p(x),
 \ee
 \tcb{with $p(x)$ given in} \eqref{iRep}. 
Combining  \eqref{7281}, \eqref{7294}, \eqref{7295}, and \eqref{7296}, we deduce that 
\be\label{7297}\ba
\liminf\limits_{n\to \infty}\frac{\E\[\int_0^\infty e^{-\rho t}U(c^n_t )dt \] - u(x)}{q_n}\geq  p(x) - \pi >0, 
\ea\ee
where the strict positivity follows from the assumption that $\pi < p(x)$.  As  $c^n\in\cA(x - q_n\pi, q_n)$, $n\in\N$,
via Definition \ref{defInd}, \eqref{7297} implies that $\pi$ is not an indifference price for $e$. By considering an income stream $\tilde e = -e$, and applying the argument above, we deduce that every $\pi>p(x)$ is not an indifference price for $e$.
\end{proof}
\section{The case where the agent has no prior knowledge about when the jumps occur, nor of their magnitudes}\label{secKnowsNothing}
\tcb{This section considers the investor who has no information about the jumps}. We proceed to recall the characterizations of the value function and the optimal consumption in \cite[Theorem 1]{Philip1}. With 
\be\label{defg}
g(q) : = \int \{ 1 + q(e^x - 1)\}^{1-R}p(x)dx,\quad q\in[0,1],
\ee
and
\be\label{defg1}
g_1(q) := r + q(\mu - r) -\frac 12 \sigma^2 Rq^2 + \frac{\lambda\{g(q) - 1\}}{1-R},\quad q\in[0,1],
\ee
we denote a maximizer to $g_1$ over $[0,1]$ by $\bar q_1$, that is 
\be\label{defq1}
\bar q_1 \in \argmax\limits_{q\in[0,1]}g_1(q),
\ee
and suppose that  {$\bar q_1\in(0,1)$}.
\begin{Remark}
The condition $\bar q_1\in(0,1)$ is needed to ensure that the dual minimizer is a local martingale. \tcb{This is in the same spirit as \cite[Theorem 3.2]{KallsenLevy}}, where a very similar assumption is imposed, ensuring that the candidate solution obtained from the first-order conditions is admissible.  
\end{Remark}
We shall also denote
\be\label{defA1}
A_1 := \(\frac{R}{\rho + (R-1)g_1(\bar q)}\)^R.
\ee
With these preliminaries, it is shown in \cite[Theorem 1]{Philip1} that, for $q=0$ and $x>0$ being fixed, if $\rho +(R-1) g_1(\bar q_1)>0$, the optimal solution to \eqref{primalProblem} has controls 
$$\bar c_t = A_1^{-\frac 1R}\bar w_t,\quad \bar\theta_t = \bar q_1\bar w_{t-},$$
where $\bar q_1$ is given by \eqref{defq1}, \tcb{$A_1$ is given by \eqref{defA1}}, and the wealth process $\bar w$ is given via 
\be\label{defw0}
 \bar w_t =\bar w_0+\int_0^t(r\bar w_{s} - \bar c_s)ds + \int_0^t\bar \theta_s\( (\mu - r)ds + \sigma  dW_s + \int(e^x - 1)n(ds, dx)\), \quad t> 0, \quad\bar  w_0 >0.
\ee
Further, the value function $u$ defined in \eqref{primalProblem} is given by 
\be\label{8281}
u(x) = A_1 U(x),\quad x>0,
\ee
with $U$ \tcb{given in} \eqref{defU}.
Now, \eqref{8281} implies that $$u(x)>-\infty, \quad x>0.$$In view of the power structure of the utility $U$ represented by  \eqref{defU} and since for every $\bar w_0>0$, we have $$\(e^{-\rho t}U'(\bar c_t)\)_{t\geq 0}\in\bigcup\limits_{y>0}\cY(y),$$ one can show that 
$$v(y) 
<\infty,\quad y>0.$$
Therefore, \eqref{finValue} holds. 
With
 $$\bar w_0 = \frac{R}{\rho + (R-1)g_1(\bar q)},$$
 and for $\bar w$ given by \eqref{defw0}, 
 one can apply \cite[Theorem 3.2]{MostovyiNec} to show that the optimal solution to \eqref{dualProblem} at $y=1$ is given by 
\be\label{hatY}
\hat Y_t = e^{-\rho t}U'(\bar c_t) = A_1 e^{-\rho t}\bar w^{-R}_t,\quad t\geq 0,\quad y>0.
\ee
This $\hat Y$ is the dual-optimal state price density. 
Theorem \ref{thmIRep} asserts that, for every income stream $e$ satisfying \eqref{superRep}, its indifference value does not depend on $x>0$ and is given by 
\be\label{iRep2}
p(x) = A_1 \E\[ \int_0^\infty e^{-\rho t}\bar w^{-R}_t e_t dt\],\quad x>0.
\ee
Invoking It\^o's lemma, one can show that $\hat Y$ satisfies the following stochastic differential equation
\be\nn\ba
{d\hat Y_t} = {\hat Y_{t-}}&\left(-\rho dt + R\{-r - \bar q_1(\mu -r ) + A_1^{-\frac 1R} + \frac 12(R+1) \sigma^2 \bar q_1^2\}dt\right.\\
&\qquad\left. -R\bar q_1 \sigma dW_t + \int \{(1  + \bar q_1(e^x - 1))^{-R}  -1\}n(dt, dx)\right),\quad t>0, \quad \hat Y_0 =1.
\ea\ee
Thus, in view of \cite[Lemma A.1]{KallsenPII}, $\hat Y B$ is a true martingale. The sufficient conditions given by Remark \ref{remYmart} thus apply.

\section{The case where the agent knows precisely the time of the next jump, but not the magnitude}\label{secKnowsTimes}
In this setting, the (first type of) insider knows the time $T_1$ of the first jump of the stock, but not the magnitude, and immediately after the jump $T_1$, he knows the time of the next jump, but not the magnitude. This pattern continues; immediately after the \tcb{jump $ T_n$}, the insider of this kind knows the time of the \tcb{jump $T_{n+1}$} but not its magnitude. \tcb{We consider} the value function
\be\label{8311}
u^{a}(t, w;T_1) = \sup\limits_{c\in\cA(w)}\E'\[ e^{-\rho(t-s)}U(c_s) ds|T_1, w_t = w\],\quad t\in [0,T),
\ee
where $\E'$ is the expectation for the insider of the first kind (here, $\P'$ is the associate probability measure on a filtered probability space, \tcb{where the filtration supports the knowledge of jump $T_{i+1}$ immediately after the jump $T_i$)}. 
Following \cite{Philip1} and using the scaling argument, one can show that the value function up to $T_1$ has the form
$$u^{a}(t, w;T_1) = f(T_1 - t)U(w),\quad t\in[0,T_1),$$
for some function $f$ to be yet characterized. Since $T_1\sim {\cE}xp(\lambda)$, we can represent $u^{a}(0,w)$ as 
$$u^{a}(0,  w) = \int \lambda e^{-\lambda s} f(s) ds U(w)=:A_2U(w),$$
where the constant $A_2$ is yet to be determined. 
At the jump time $T_1$, if the wealth of the insider of the first kind will jump from 
$$w_{T_1-}\quad to \quad w_{T_1-}\(1 + a (e^\xi - 1)\),$$
where $a$ is the proportion of the investor's wealth invested in the risky asset at $T_1-$ and $\xi\sim N(m, v)$,  \tcb{the utility of his wealth} will get scaled by a factor $\(1 + a (e^\xi - 1)\)^{1-R}$. Thus we choose $$a^*\in argmax_{a\in[0,1]} \frac{g(a)}{1- R}.$$
\tcb{Similar to the case} of the investor without insider information, \tcb{in order to ensure} that the candidate dual minimizer $\hat Y^b$ is such that $e^{rt}\hat Y^b_t$, $t\geq 0$, is a local martingale  and not a supermartingale, we suppose that there exists a unique 
\be\label{a*}
a^*\in argmax_{a\in[0,1]} \frac{g(a)}{1- R}\bigcap (0,1).
\ee
The authors of \cite{Philip1} show that with
$$\gamma_M: = \frac{\rho + (R-1)\(r + \frac{(\mu -r)^2}{2R\sigma^2}\)}R,$$ 
\be\label{deff1}
f(t) = \(\frac{1 - e^{-\gamma_M t}}{\gamma_M} +e^{-\gamma_Mt}f(0)^{\frac 1R} \)^R,\quad t\geq 0,\ee
where $f(0)$ is the unique fixed point of 
\be\label{deff0}x\to \varphi(x): = g(a^*)\int_0^\infty \lambda e^{-\lambda s}\( \frac {1 -  e^{-\gamma_Ms}}{\gamma_M} + e^{-\gamma_M s}x^{\frac 1R}\)^Rds,
\ee
which, by \cite[Theorem 2]{Philip1}, exists if $R>1$ or if $R<1$ and $\frac{\lambda g(a^*)}{\lambda + R\gamma_M}< 1$. Denoting $T_0:=0$, between the jumps, it is optimal to have 
$$\theta^{a}_t = \frac{\mu - r}{\sigma^2 R}w^{a}_{t-},\quad   c^{a}_t = f(T_i-t)^{-\frac 1R}w^{a}_t, \quad t\in (T_{i-1},T_i),\quad i\in\N,$$ 
The constant $A_2$ is given by 
$$A_2 = \int \lambda e^{-\lambda s} f(s) ds.$$
Note that the problem \eqref{primalProblem} is ill-posed if $R\in(0,1)$ and 
$$\frac{\lambda g(a^*)}{\lambda + R\gamma_M}\geq 1,$$
which is consistent with \cite[Proposition 1.3]{RogersBook}. Otherwise, (that is, if $R>1$ or if $R\in(0,1)$ and $\frac{\lambda g(a^*)}{\lambda + R\gamma_M}< 1$), on $[0,T_1]$, \tcb{using} the convention that $w^{a}_{0-} = w^{a}_0$, 
 the optimal controls are given by 
$$\theta^{a}_t = \frac{\mu - r}{\sigma^2 R} w^{a}_{t-},\quad t\in[0, T_1),\quad   \theta^{a}_{T_1} = a^*    w^{a}_{T_1-}, \quad   c^{a}_t = f(T_1-t)^{-\frac 1R}    w^{a}_{t-}, \quad t\in[0, T_1],$$
\tcb{where the time $T_1$ is the renewal time}. The solution after $T_1$ conforms with the solution stated above for $[0,T_1]$ throughout the interval $[T_1,T_2]$ and recursively thereafter. That is, we have, respectively, the following optimal strategy and consumption.
$$   \theta^{a}_t =    w^{a}_{t-}\sum\limits_{i = 1}^{\infty}\(\frac{\mu - r}{\sigma^2 R} 1_{(T_{i-1}, T_i)}(t)+ a^* 1_{[T_i]}(t)\),\quad  c^{a}_t =   w^{a}_{t-}\sum\limits_{i = 1}^{\infty} f(T_i - t)^{-\frac 1R}1_{[T_{i-1}, T_i)}(t),\quad t\geq 0,$$
\be\label{defw1}
   w^{a}_t =   w^{a}_0  + \int_0^t(r  w^{a}_s -   c^{a}_s)ds + \int_0^t  \theta^{a}_s\( (\mu - r)ds + \sigma  dW_s + \int(e^x - 1)n(ds, dx)\), \quad t> 0,\quad   w^{a}_0>0.
\ee
If we consider the initial wealth  {$w^{a}_0 = \(f(T_1)\)^{\frac 1 R}$}, and, e.g., by 
invoking \cite[Theorem 3.2]{MostovyiNec} and \cite[Appendix A.2]{Philip1}, one may show that 
\be\label{dualOptI1}\hat Y^{a}_t = \sum\limits_{j = 1}^{\infty} f(T_j - t)\(w^{a}_t\)^{-R}e^{-\rho t}1_{[T_{j-1}, T_j)}(t),\quad t\geq 0,\ee
is the optimizer to \eqref{dualProblem} at $y=1$ for the insider of the first kind. That is, $\hat Y^{a}$ is the dual-optimal  state price density, which satisfies
\be\nn\hat Y^{a}_t = e^{-\rho t}U'(c^{a}_t),\quad t\geq 0,\quad (dt\times \P')\text-a.e.,\ee
and where $c^{a}$ is the optimizer to \eqref{primalProblem} for the insider of the first kind at $w^{a}_0 = \(f(T_1)\)^{\frac 1 R}$.

 Therefore, to compute the indifference value of an income stream $e$ for an insider of this kind, we need to invoke the representation formula given by \eqref{iRep} in Theorem \ref{thmIRep}, with the dual optimizer corresponding to $y=1$ being given by \eqref{dualOptI1}.
 
%

\section{The case where the agent knows nothing about the time of the next jump but sees the signal
$\eta = \xi + \varepsilon $} \label{secKnowsSize}
The insider of the second kind has some advance information about the size of the first jump of the Poisson random measure $n$ that comes at time $T_1$, \tcb{ but knows nothing about the time of this jump}.  This insider receives a signal 
$$\eta = \xi + \e,$$
where $\xi\sim N(0, v_\e)$ is independent of everything else. One can compute the conditional distribution of $\xi$ given $\eta$ as
\be\label{921}
(\xi|\eta)\sim N\(\frac{v\eta + v_\e m}{v+ v_\e},\frac {vv_\e}{v + v_\e}\).
\ee
At time $T_1$, immediately after the jump, the investor receives the next jump similar to the one above. 
With $E''$ denoting the expectation for this kind of insider, we consider 
$$u^{b}(w;\eta):=\sup\limits_{c\in\cA(w)}\E''\[ \int_0^\infty e^{-\rho t}U(c_t)dt|\eta, w_0 = w\].$$
Using the scaling argument, one can show that 
$$u^{b}(w;\eta) = h(\eta) U(w),$$
for a function $h$ \tcb{yet to be} characterized. If at $T_1$, $\xi$ is the size of the jump in the log price and $q$ is the fraction of wealth invested in the risky asset, the wealth of the insider of the second kind changes by the multiplicative factor $1 + q(e^\xi - 1)$, so that the indirect utility of the insider changes at $T_1$ from $h(\eta)U(w)$ to 
$$\(1 + q(e^\xi - 1)\)^{1-R}U(w)h(\eta'),$$
where $\eta'$ is the new signal about the size of the jump at time $T_2$. With $\P_0$ denoting the $N(m, v + v_\e)$ distribution and $\P$ denoting the distribution in \eqref{921}, the expected value of the indirect utility at $T_{1+}$ is 
$$U(w) \int \(1 + q(e^x - 1)\)^{1-R} \P(dx|\eta)\int h(y)\P_0(dy).$$
Denoting 
$$A_3: = \int h(y) \P_0(dy),$$
the expected value in the indirect utility can be written as 
$$A_3U(w) \int \(1 + q(e^x - 1)\)^{1-R} \P(dx|\eta).$$ 
Thus, determining $A_3$ and $h$ is key for finding the dual-optimal state price density for \tcb{this insider of the second kind}. With $\P(dx|m', v')$ denoting \tcb{the} $N(m', v')$ distribution, we introduce
\be\label{defphis}\ba
\varphi_1(q) &: = r + q (\mu -r) -\frac 12 \sigma^2 q^2 R - \frac{\rho + \lambda}{1 -R},\\
\varphi_2(q;m', v') &: = \int U(1 + q(e^x - 1))\P(dx|m', v'),\quad q\in[0,1].
\ea
\ee
If $R>1$ and $\frac{\mu - r}{\sigma^2R}\in(0,1)$,  it is shown in \cite[Theorem 3]{Philip1} that the pair $(h, A_3)$ is given by a maximal solution to the system
\be\label{9191}\ba
\frac{h(\eta)^{1 - \frac 1R}}{1 - \frac 1R} &= \sup\limits_{0\leq q\leq 1}\( h(\eta) \varphi_1(q) + \lambda A_3 \varphi_2(q; m'(\eta), v')\),
\\
A_3 &= \int h(y) \P_0(dy),
\ea\ee
subject to $A_3\leq A_1$. Then, the value function $u^{b}(w;\eta) = h(\eta)U(w)$, and the optimal $q$ denoted by $\bar q(\eta)$, maximizes the function 
\be\label{922}
q\to h(\eta) \psi_1(q) + \lambda A_3 \varphi_2\(q;\frac{v\eta + v_\e m}{v+ v_\e},\frac {vv_\e}{v + v_\e}\),\quad q\in[0,1].
\ee
With $w_{0-} := w_{0}$, the optimal controls are given by 
\be\nn
c^{b}_t = (h(\eta_t))^{-\frac 1R} w^{b}_{t\tr -},\quad \theta^{b}_t = \bar q(\eta_t)w^{b}_{t\tr -},\quad t\geq 0,
\ee
where $\eta_t$ is the signal known at time $t$ about the next jump after $t$. The optimal wealth process is given by 
\be\label{defw2}
 w^{b}_t =   w^{b}_0  + \int_0^t(r  w^{b}_s - c^{b}_s)ds + \int_0^t \theta^{b}_s\( (\mu - r)ds + \sigma  dW_s + \int(e^x - 1)n(ds, dx)\), \quad t> 0,\quad w^{b}_0>0.
\ee
Now, using \cite[Theorem 3]{Philip1} and \cite[Theorem 3.2]{MostovyiNec}, with \tr{$w^{b}_0 = (h(\eta_0))^{\frac 1R}$}
, one can show that the dual optimal state price density is given by 
$$\hat Y^{b}_t = e^{-\rho t} h(\eta_t) \(w^{b}_t\)^{-R},\quad t\geq 0,$$
where, again, $\eta_t$ is the signal known at time $t$ about the next jump after $t$.

\section{Discussion}\label{secDiscussion}
We begin by \tcb{considering the case of the investor who has no information about the jumps}. If the return of the risky asset is given by 
$$R = \bar \mu t + \bar \sigma W,$$
and the riskless asset gives a constant return $rt$ as in the previous sections, 
as in the Black-Scholes-Merton case, one can find the value functions in \eqref{primalProblem} and \eqref{dualProblem}  and the optimizers to these problems. Here, we have that
$$u^M(w) = A_MU(w),$$
where $$A_M = \frac{R^R}{\(\rho + (R-1)\(r +\frac{(\bar \mu-r)^2}{2\bar \sigma^2 R}\)\)^R}.$$
Following \cite[Chapter 1.2]{RogersBook}, with $$\kappa := \frac{\mu -r}{\sigma}\quad{\rm and}\quad \gamma_M : = A_M^{-\frac 1R},$$ we have the following optimal wealth process $$w^M_t = w^M_0\exp\( R^{-1}\kappa\cdot W_t + (r + \frac 12 R^{-2}|\kappa|^2(2R - 1) -\gamma_M)t\),\quad t\geq 0.$$
The optimal consumption rate is $c^M:=\gamma_Mw^M$, and therefore, for $w^M_0 =\frac 1{\gamma_M}$, the dual-optimal state price density is 
$$\hat Y^M_t = e^{-\rho t}U'(c^M_t) = e^{-\rho t}\(w^M_t\)^{-R},\quad t\geq 0,$$
\tcb{and the value function in \eqref{primalProblem} is $u^M(w) = A_MU(w)$.} Observe that \eqref{finValue} holds.
\tcb{We now compare the optimal state-price densities in previous sections}:
$$\hat Y_t = e^{-\rho t}A_1\bar w^{-R}_t,\quad u(w) = A_1U(w),$$
$$\hat Y^{a}_t= e^{-\rho t}\sum\limits_{i\geq 1}f(T_i - t)  \(w^{a}_t\)^{-R}_t1_{[T_{i-1}, T_i)}(t),\quad u^{a}(w) = A_2U(w),$$
$$\hat Y^{b}= e^{-\rho t}h(\eta_t) \(w^{b}_t\)^{-R},\quad u^{b}(w) = A_3U(w).$$

\begin{Example}

For a deterministic income stream, for example,  $e _t = 1$, $t\geq 0$, the indifference value does not change under extra information. \tcb{Its indifference value is $\frac 1r$ under both the choice of parameters in Section \ref{secKnowsTimes} and under the choice of parameters in Section \ref{secKnowsSize}}.
\end{Example}

\begin{Example}
Let us consider a deterministic yet randomly terminating income stream $e_t = \exp(rt)1_{[0,T_1)(t)}$, $t\geq 0$. 
It is easiest to compute the indifference value of this stream for the insider of the first kind. Here, the value of the stream given $T_1$ is 
$$  \E'\[\int_0^\infty e_t \hat Y^a_t dt|T_1\] =  \E'\[\int_0^{T_1} e^{rt} \hat Y^a_t dt|T_1\] = T_1,$$
and therefore, on average, the value of such a stream is $\frac 1{\lambda}$.

Next, for the \tcb{investor who has no information about the jumps}, with $A_1$ given in \eqref{defA1} and $\bar q_1$ the unique maximizer in \eqref{defq1}, where we suppose that $\bar q_1\in(0,1)$, set
\be\label{defAlpha}
\alpha : = 
r-\rho + R\left\{-r - \bar q_1(\mu -r ) + A_1^{-\frac 1R} + \frac 12(R+1) \sigma^2 \bar q_1^2\right\}.
\ee
The value of the stream is given by 
\be\nn\ba 
\E\[\int_0^\infty e_t \hat Y_t dt \] &=   \E\[\int_0^\infty e^{rt} \hat Y_t1_{[0,T_1)(t)} dt\] \\
&= \E\[\int_0^\infty \exp\(\alpha t - R\bar q_1 \sigma W_t - \frac 12\(R \bar q_1 \sigma\)^2 t \)1_{[0,T_1)(t)} dt\] \\
&= \int_0^\infty \E\[\exp\(\alpha t - R \bar q_1 \sigma W_t- \frac 12\(R \bar q_1 \sigma\)^2 t \)1_{[0,T_1)(t)}\] dt \\
&= \int_0^\infty \E\[\exp\(\alpha t - R \bar q_1 \sigma W_t- \frac 12\(R \bar q_1 \sigma\)^2 t \)\]\P[T_1>t] dt \\
&= \int_0^\infty \E\[\exp\(\alpha t - R \bar q_1 \sigma W_t- \frac 12\(R \bar q_1 \sigma\)^2 t \)\]\exp(-\lambda t) dt \\
&= \int_0^\infty \E\[\exp\( - R \bar q_1 \sigma W_t- \frac 12\(R \bar q_1 \sigma\)^2 t \)\]\exp\(\{\alpha -\lambda\} t\) dt \\
& = \int_0^\infty \exp\(\{\alpha -\lambda\} t\) dt = \frac{1}{\lambda-\alpha},
\ea\ee
\tcb{which,  if $\alpha>0$,} is greater than $\frac 1\lambda$ (here, $\frac 1\lambda$ is the average value of the insider of the first kind) and is smaller than $\frac 1\lambda$ if $\alpha<0$.

For the insider of the second kind, similar to the computations for \tcb{the investor who has no information about the jumps}, with $\eta_0$ being the initial signal about the first jump (at $T_1$) given to this insider, $\bar q(\eta_0)$ being the $\omega$-by-$\omega$ maximizer to \eqref{922} associated with $\eta_0$, and $h$ being given through a pair $(h, A_3)$ as a maximal solution to \eqref{922}, set
\be\label{defBeta}\beta(\eta_0): = r-\rho + R\left\{ -r - \bar q(\eta_0) (\mu -r) + h(\eta_0)^{-\frac 1R} +\frac 12(R+1)\sigma^2\(\bar q(\eta_0)\)^2\right\}.
\ee
\tcb{We compute}
\be\nn\ba
\E''\[\int_0^\infty e_t \hat Y^{b}_tdt |\eta_0\] &= \E''\[\int_0^\infty e^{rt}\hat Y^{b}_t1_{[0,T_1)}(t)dt |\eta_0\]\\
&= \int_0^\infty \E''\[\exp\(-\bar q(\eta_0) R \sigma W_t -\frac 12 (\bar q(\eta_0) R \sigma)^2t \)|\eta_0\]\exp\(\beta(\eta_0) t  \)\P\[T_1>t\]dt \\
&= \int_0^\infty \exp\(\{\beta(\eta_0) - \lambda\} t  \)dt  = \frac 1{\lambda - \beta(\eta_0)}.
\ea\ee
\end{Example}
\begin{Example}
Consider an income stream $e_t = e^{(r-1)t} 1_{[T_1,\infty)}(t)\Psi(\eta_0)$, $t\geq 0$, where $\Psi $ is a bounded function. That is, the income stream is set at $T_1$; it depends on both $T_1$ and the signal that the investor of the second kind receives at time $0$ about the first jump of the risky asset, $\eta_0$. \tcb{After $T_1$, the income stream is received indefinitely by a predetermined formula at $T_1$.}

For the \tcb{investor who has no information about the jumps}, with the clock $\kappa_t := 1 - e^{-t}$, $t\geq 0$, we obtain
\be\nn\ba
\E\[\int_0^\infty e_t \hat Y_t dt\] &= \E\[\int_{T_1}^\infty \(e^{rt}\hat Y_t\) \Psi(\eta_0) e^{-t}dt\]\\
&= \E\[\int_{T_1}^\infty \(e^{rt}\hat Y_t\) \Psi(\eta_0) d\kappa_t\]\\
&= \E\[\E\[\int_{T_1}^\infty \(e^{rt}\hat Y_t\) \Psi(\eta_0) d\kappa_t|\cF_{T_1}\]\],
\ea\ee
which, using localization and integration by parts, we can rewrite as
\be\nn\ba
\E\[ \Psi(\eta_0)  \(e^{rT_1}\hat Y_{T_1}\) e^{-T_1}\].
\ea\ee
Recalling that $\eta_0 = \xi +\e$, where $\xi\sim N(m, v)$ and $\e\sim N(0, v_\e)$ are independent (from each other and the Brownian motion $W$), we can rewrite the latter expression as 
\be\label{9201}\ba
\E\[ \Psi(\xi + \e) \(e^{r{T_1-}}\hat Y_{{T_1-}}\) \(1 + \bar q_1\(e^{\xi} - 1\)\)^{-R}e^{-T_1}\].
\ea\ee
where $\bar q_1$ is the unique maximizer in \eqref{defq1}, where we suppose that $\bar q_1\in(0,1)$. Finally, with $\alpha$ defined in \eqref{defAlpha}, and assuming that $\alpha<1+\lambda$, we restate the value of the income stream given by \eqref{9201} below
 \be\nn\ba
&\int_0^\infty e^{-(1 - \alpha)t} \lambda e^{-\lambda t}dt \int_{\R^2} \frac 1{2\pi\sqrt{vv_e}} \Psi(x_1 + x_2)  \(1 + \bar q_1\(e^{x_1} - 1\)\)^{-R} e^{-\frac{(x_1 - m)^2}{2v} - \frac{x_2^2}{2v_e}}dx_1dx_2\\
=&
\frac{\lambda}{\(\lambda +1-\alpha\)} 
\frac 1{2\pi\sqrt{vv_e}}
\int_{\R^2}  \Psi(x_1 + x_2)  \(1 + \bar q_1\(e^{x_1} - 1\)\)^{-R} e^{-\frac{(x_1 - m)^2}{2v} - \frac{x_2^2}{2v_e}}dx_1dx_2.
\ea\ee

Now, let us consider an insider of the second kind who knows the time of every jump. \tcb{Similar to the computations above (particularly  \eqref{9201})}, with 
$ {\pi_M := \frac{\mu - r}{\sigma^2 R}}$, \tcb{with} $f$ being given through \eqref{deff1} and \eqref{deff0}, and 
$$\phi(T_1) :=r-\rho  +R\(-r - \pi_M(\mu - r)+1_{\{T_1>0\}}\frac{ 1}{T_1}\int_0^{T_1}f(T_1-t)^{-\frac 1R}dt+ \frac 12 \pi^2_M\sigma^2(1 + R)\) ,$$
we obtain the following value of the income stream for the insider of the first kind
\be\nn\ba
&\E'\[ \Psi(\xi + \e)\(e^{r{T_1-}}\hat Y^a_{{T_1-}}\)e^{-T_1}\(1 + a^*\(e^{\xi} - 1\)\)^{-R}|T_1\]\\
=&e^{-T_1}\E'\[\(e^{r{T_1-}}\hat Y^a_{{T_1-}}\)\Psi(\xi + \e) \(1 + a^*\(e^{\xi} - 1\)\)^{-R}|T_1\]\\
=&e^{-T_1}
\E'\[\(e^{r{T_1-}}\hat Y^a_{{T_1-}}\)|T_1\]
\frac 1{2\pi\sqrt{vv_e}}
\int_{\R^2}  \Psi(x_1 + x_2)  \(1 + a^*\(e^{x_1} - 1\)\)^{-R} e^{-\frac{(x_1 - m)^2}{2v} - \frac{x_2^2}{2v_e}}dx_1dx_2\\
=&e^{-T_1}\frac{f(0)}{f(T_1-)}e^{\phi(T_1)T_1 }
\frac 1{2\pi\sqrt{vv_e}}
\int_{\R^2}  \Psi(x_1 + x_2)  \(1 + a^*\(e^{x_1} - 1\)\)^{-R} e^{-\frac{(x_1 - m)^2}{2v} - \frac{x_2^2}{2v_e}}dx_1dx_2.
\ea\ee

Finally, consider the insider of the second kind, who receives the signal about the jump at $T_1$, $\eta_0 = \xi + \e$, where $\xi$ is the actual size of the jump of the cumulative return of the risky asset at $T_1$, \tcb{and where $\e$ is an independent $N(0, v_\e)$ random variable assumed to be independent of everything else.}
 \tcb{Proceeding similarly} to the computations for the \tcb{investor who has no information about the jumps}, we obtain the following value of the income stream
 \be\nn\ba
\E''\[ \Psi(\eta_0)  \(e^{rT_1}\hat Y_{T_1}\) e^{-T_1}|\eta_0\] &=  \Psi(\eta_0) \E''\[ \(e^{rT_1}\hat Y_{T_1}\) e^{-T_1}|\eta_0\] \\
 &=  \Psi(\xi + \e) \E''\[ \(e^{rT_1}\hat Y_{T_1-}\) e^{-T_1}\(1 + \bar q(\eta_0)(e^\xi - 1) \)^{-R}|\eta_0\]. \\
\ea\ee
 \tcb{Recall that} $\beta(\eta_0)$ is defined in \eqref{defBeta}. Assuming that $\beta(\eta_0)<1+ \lambda$, and observing that the conditional distribution of $\xi$ given $\eta_0$ is $N\(\frac {v\eta_0 + v_\e m}{v+v_\e}, \frac {vv_\e}{v+v_\e}\)$, 
 we can simplify the latter expression as 
  \be\nn\ba
  &\Psi(\eta_0) \E''\[ e^{(\beta(\eta_0)-1)T_1}e^{\(-\bar q(\eta_0) R\sigma W_{T_1} -\tfrac 12  (\bar q(\eta_0) R\sigma)^2T_1\)}\(1 + \bar q(\eta_0)(e^\xi - 1) \)^{-R}|\eta_0\]\\
  =&\Psi(\eta_0) \int_0^\infty \lambda e^{-\lambda t}e^{(\beta(\eta_0)-1)t}\E''\[e^{\(-\bar q(\eta_0) R\sigma W_{t} -\tfrac 12  (\bar q(\eta_0) R\sigma)^2t\)}\(1 + \bar q(\eta_0)(e^\xi - 1) \)^{-R}|\eta_0\]dt\\
  =&\Psi(\eta_0) \frac \lambda{\lambda + 1 - \beta(\eta_0)}\frac {\sqrt{vv_\e}}{\sqrt{2\pi (v + v_\e)}}\int_{\R}e^{-\frac{\(x - \frac{v\eta_0 + v_\e m}{v+ v_\e}\)^2(v + v_\e)}{2vv_\e}}\(1 + \bar q(\eta_0)(e^{x}- 1) \)^{-R} dx.
\ea\ee
\end{Example}

\bibliographystyle{alpha}\bibliography{finance}
\end{document}